\documentclass{llncs}
\usepackage{amsmath, amssymb, stmaryrd, bbm,
bussproofs,xspace,wrapfig}
\usepackage{multicol}
\usepackage[all]{xy}
\SelectTips{cm}{}
\usepackage{textcomp}
\usepackage{times}
\usepackage{relsize}
\usepackage[final]{graphicx}

\usepackage[author=anonymous,nomargin,marginclue,footnote,status=draft]{fixme}
\FXRegisterAuthor{dp}{adp}{DP}
\FXRegisterAuthor{ak}{aak}{AK}
\FXRegisterAuthor{ls}{als}{LS}
\FXRegisterAuthor{sm}{asm}{SM}

\newlength{\myboxwidth}
	{\end{center}\end{figure}}

\usepackage{algorithm}
\usepackage[final]{listings}

\newcounter{blubber}

\lstdefinelanguage{pseudocode}{
  basicstyle=\footnotesize,
  numbers=left,
  numberstyle=\footnotesize,
  stepnumber=1,
  numbersep=2pt,                  
  breaklines=false,                
  breakatwhitespace=false,        
  morekeywords={def,input,output,where,let,while,for,to,invariant,return,do,else},
  morecomment=[l]{//},
  escapechar=\&,
  literate={:=}{{$\gets$}}1 {+=}{{$\modifyWith{+}$}}1
}

\usepackage{lipsum}

\def\Id{\mathit{Id}}
\def\refeq#1{(\ref{#1})}
\def\Powf{\Pow_{<\mathsf{\omega}}}

\newcommand{\yes}{\mathsf{y}}
\newcommand{\no}{\mathsf{n}}

\renewcommand{\theta}{\vartheta}

\newcommand{\argument}{\_\!\_}
\newcommand{\bang}{{!}}

\newcommand{\Kl}{\mathsf{Kl}}
\newcommand{\id}{\mathit{id}}

\newcommand{\BC}{\mathbf{C}}

\newcommand{\Pow}{\mathcal{P}}
\newcommand{\fD}{\+D} 
\newcommand{\Set}{\mathsf{Set}}

\newcommand{\FA}{\mathfrak{A}}

\makeatletter
\newcommand*{\@old@slash}{}\let\@old@slash\slash
\def\slash{\relax\ifmmode\delimiter"502F30E\mathopen{}\else\@old@slash\fi}
\makeatother

\spnewtheorem{thm}[theorem]{Theorem}{\bfseries}{\itshape}
\spnewtheorem{cor}[theorem]{Corollary}{\bfseries}{\itshape}
\spnewtheorem{lem}[theorem]{Lemma}{\bfseries}{\itshape}
\spnewtheorem{lemdefn}[theorem]{Lemma and Definition}{\bfseries}{\itshape}
\spnewtheorem{propn}[theorem]{Proposition}{\bfseries}{\itshape}
\spnewtheorem{defn}[theorem]{Definition}{\bfseries}{\upshape}
\spnewtheorem{obs}[theorem]{Observation}{\bfseries}{\upshape}
\spnewtheorem{rem}[theorem]{Remark}{\bfseries}{\upshape}
\spnewtheorem{expl}[theorem]{Example}{\bfseries}{\upshape}
\spnewtheorem{thmdefn}[theorem]{Theorem and Definition}{\bfseries}{\itshape}
\spnewtheorem{propdefn}[theorem]{Proposition and Definition}{\bfseries}{\itshape}
\spnewtheorem{assn}[theorem]{Assumption}{\bfseries}{\upshape}
\spnewtheorem{conv}[theorem]{Convention}{\bfseries}{\upshape}
\spnewtheorem{notn}[theorem]{Notation}{\bfseries}{\upshape}
\spnewtheorem{open}[theorem]{Problem}{\bfseries}{\upshape}

\begin{document}

\title{Simplified Coalgebraic Trace Equivalence} \author{Alexander Kurz\inst{1}\and Stefan Milius\inst{3}\and Dirk Pattinson\inst{2} \and Lutz Schr{\"o}der\inst{3}
} \institute{University of Leicester
\and
The Australian National University
\and
Friedrich-Alexander-Universit\"at
  Erlangen-N{\"u}rnberg}
\maketitle

\begin{abstract}
  The analysis of concurrent and reactive systems is based to a large
  degree on various notions of process equivalence, ranging, on the
  so-called linear-time/branching-time spectrum, from fine-grained
  equivalences such as strong bisimilarity to coarse-grained ones such
  as trace equivalence. The theory of concurrent systems at large has
  benefited from developments in coalgebra, which has enabled uniform
  definitions and results that provide a common umbrella for seemingly
  disparate system types including non-deterministic, weighted,
  probabilistic, and game-based systems. In particular, there has been
  some success in identifying a generic coalgebraic theory of
  bisimulation that matches known definitions in many concrete
  cases. The situation is currently somewhat less settled regarding
  trace equivalence. A number of coalgebraic approaches to trace
  equivalence have been proposed, none of which however cover all
  cases of interest; notably, all these approaches depend on explicit
  termination, which is not always imposed in standard systems,
  e.g. LTS. Here, we discuss a joint generalization of these
  approaches based on embedding functors modelling various aspects of
  the system, such as transition and braching, into a global monad;
  this approach appears to cover all cases considered previously and
  some additional ones, notably standard LTS and probabilistic
  labelled transition systems.
\end{abstract}

\section{Introduction}

It was recognized early on that the initial algebra semantics of Goguen and Thatcher \cite{GoguenT74} needs to be extended to account for notions of observational or behavioural equivalence, see Giarratana, Gimona and Montanari \cite{GGM76}, Reichel \cite{Reichel81}, and Hennicker and Wirsing \cite{HW85}. When Aczel \cite{aczel:nwfs} discovered that at least one important notion of behavioural equivalence---the bisimilarity of process algebra---is captured by final coalgebra semantics, the study of coalgebras entered computer science. Whereas early work emphasized the duality between algebra and coalgebra, it became soon clear that both areas have to be taken together. For example, in the work of Turi and Plotkin \cite{TuriP97}, monads represent the programs, comonads represent their behaviour (operational semantics), and a distributive law between them ensures that the behaviour of a composed system is given by the behaviours of the components, or, more technically, that bisimilarity is a congruence.

Another example of the interplay of algebraic and coalgebraic structure arises from the desire to make coalgebraic methods available for a larger range of program equivalences such as described in van Glabbeek's \cite{Glabbeek90}. To this end, Power and Turi \cite{PowerT99} argued that trace equivalence arises from a distributive law $TF\to FT$ between a monad $T$ describing the non-deterministic part and a functor $F$ describing the deterministic part of a transition system $X\to TFX$.
This was taken up by Hasuo et al \cite{HasuoEA07} and gave rise to a
sequence of papers \cite{KissigKurz10,JacobsEA12,sbbr13,bms13,Cirstea:2014:CAL}
that discuss coalgebraic aspects of trace equivalence.

We generalize this approach and call a trace semantics for coalgebras
$X\to GX$ simply a natural transformation $G\to M$ for some monad
$M$. This allows us, for example, and opposed to the work cited in the
previous paragraph, to account for non-determinstic transition
systems without explicit termination. Moreover, because of the
flexibility afforded by choosing $M$, both trace semantics and
bisimilarity can be accounted for in the same setting. We also show
that for $G$ being of the specific forms investigated in
\cite{HasuoEA07} and in \cite{sbbr13,bms13,JacobsEA12} there is a uniform way of
constructing the a natural transformation of type $G\to M$ that
induces  the traces of \emph{op.cit.} up to canonical forgetting of
deadlocks.

\section{Preliminaries}\label{sec:prelim}
We work with a base category $\BC$, which we may assume for simplicity to be locally finitely presentable, such as the category $\Set$ of sets and functions.

Given a functor $G:\BC\to\BC$, a $G$-\emph{coalgebra} is an arrow $\gamma:X\to GX$. Given two coalgebras $\gamma:X\to GX$ and $\gamma':X'\to GX'$, a \emph{coalgebra morphism} $f:(X,\gamma)\to(X',\gamma')$ is an arrow $f:X\to X'$ in $\BC$ such that $\gamma'\circ f=Gf\circ\gamma$. 

When $\BC$ is a concrete category, we say that two states $x\in X$ and $x'\in X'$ in two coalgebras $(X,\gamma)$ and $(X',\gamma')$ are \emph{behaviourally equivalent} if there are coalgebra morphisms $f,f'$ with common codomain $(Y,\delta)$ such that $f(x)=f'(x')$.

Behavioural equivalence can be computed in a partition-refinement style using the \emph{final coalgebra sequence} $(G^n1)_{n<\omega}$ where $1$ is a final object in $\BC$ and $G^n$ is $n$ fold application of $G$. The projections $p^{n+1}_n:G^{n+1}1\to G^n1$ are defined by induction where $p^1_0:G\to 1$ is the unique arrow to 1 and $p^{n+2}_{n+1} = G(p^{n+1}_{n})$.  

For any coalgebra $(X,\gamma)$, there is a \emph{canonical cone}
$\gamma_n:X\to G^n1$ defined inductively by $\gamma_0:X\to 1$ and
$\gamma_{n+1}=G(\gamma_n)\gamma$. We say that two states $x,x'\in X$
in $(X,\gamma)$ are \emph{finite-depth behaviourally equivalent} if
$\gamma_n(x)=\gamma_n(x')$ for all $n<\omega$. (We remark that if $G$
is a finitary set functor, then finite-depth behavioural equivalence implies
behavioural equivalence.)

A \emph{monad} is given by an operation $M$ on the objects of $\BC$ and, for each set $X$, a function $\eta_X:X\to MX$ and, for each $f:X\to MY$, a so-called Kleisli star $f^*:MX\to MY$ satisfying
(i) $\eta_X^* = \id_{MX}$, 
(ii) $f^*\circ\eta_X = f$, 
(iii) $(g^*\circ f)^* = g^*\circ f^*$
for all $g:Y\to MZ$. It follows that $M$ is a functor, given by
$Mf=(\eta f)^*$, and $\eta$ a natural transformation. Moreover,
$\mu=\id^*:MM\to M$ is a natural transformation and satisfies
$\mu\circ M\eta=\mu\circ \eta M=\id$ and $\mu\circ M\mu=\mu\circ
\mu M$. We obtain the Kleisli star back from $\mu$ and $M$ by
$f^*=\mu Mf$.

An \emph{Eilenberg-Moore algebra} for the monad $M$ is an arrow $\xi:MX\to X$ such that $\xi\circ\eta_X=\id_X$ and $\xi\circ M\xi=\xi\circ \mu_X$.

Recall that an endofunctor $G$ on a category $\BC$ is said to generate
an \emph{algebraically-free} monad $G^*$ if the category of
Eilenberg-Moore algebras of $G^*$ is isomorphic over $\BC$ to the
category of $G$-algebras (i.e.\ morphisms $GX\to X$). The
monad $G^*$ is then also the free monad over $G$; conversely, free
monads are algebraically-free if the base category $\BC$ is
complete~\cite{Barr70,Kelly80}. E.g., when $\BC$ is locally finitely
presentable, then every finitary functor on $\BC$, representing a type
of finitely-branching systems, generates an (algebraically-)free
monad.

\section{A Simple Definition of Coalgebraic Trace Equivalence}

Recall the classical distinction between bisimilarity and trace
equivalence, the two ends of the \emph{linear-time-branching time
  spectrum}~\cite{Glabbeek90}: to cite a much-belaboured standard
example, the two labelled transition systems (over the alphabet
$\Sigma=\{a,b,c\}$)
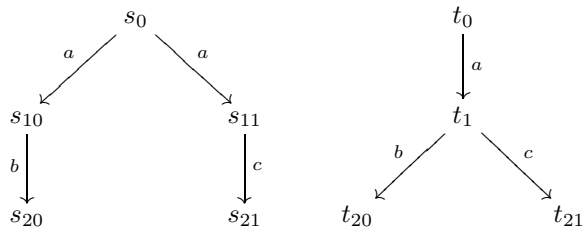
\begin{equation*}
  \xymatrix{& s_0\ar[dl]_a \ar[dr]^a & & & t_0 \ar[d]^a\\
    s_{10}\ar[d]_b & & s_{11}\ar[d]^c & & t_1\ar[dl]_b\ar[dr]^c &  \\
    s_{20} & & s_{21} & t_{20}  & & t_{21}
 }
\end{equation*}
are \emph{trace equivalent} in the usual sense~\cite{AcetoEA07}, as
they both admit exactly the traces $ab$ and $ac$ (and prefixes
thereof), but not bisimilar, as bisimilarity is sensitive to the fact
that the left hand side decides in the first step whether $b$ or $c$
will be enabled in the second step, while the right hand side leaves
the decision between $b$ and $c$ open in the first step. In other
words, trace equivalence collapses all future branches, retaining only
the branching at the current state. Now observe that we can
nevertheless construct the trace semantics by stepwise unfolding; to
do this, we need to a) remember the last step reached by a given trace
in order to continue the trace correctly, and b) implement the
collapsing correctly in each step. E.g.\ for $s_0$ above, this takes
the following form: let us call a pair $(u,x)$ consisting of a word
over $\Sigma$ and a state $x$ a \emph{pretrace}. Before the first
step, we assign, by default, the set $\{(\epsilon,s_0)\}$ of
pretraces, where $\epsilon$ denotes the empty word. After the first
step, we reach, applying both transitions simultaneously, the set
$\{(a,s_{10}),(a,s_{11})\}$. After the second step, we reach, again
applying two transitions, $\{(ab,s_{20}), (ac,s_{21})\}$. Note that
after the third step, the set of pretraces will become empty if we
proceed in the same manner, as $s_{20}$ and $s_{21}$ are both
deadlocks. Thus, we will in general need to remember all finite
unfoldings of the set of pretraces, as traces ending in deadlocks will be
lost on the way. Of course, for purposes of trace equivalence we are
no longer interested in the states reached by a given trace, so we
forget the state components of all pretraces that we have accumulated,
obtaining the expected prefix-closed trace set $\{\epsilon,a,ab,ac\}$.

Recall that we can understand labelled transition systems as
coalgebras $\gamma:X\to \Pow(\Sigma\times X)$. What is happening in
the unfolding steps is easily recognized as composition with $\gamma$
in the Kleisli category of a suitable monad, specifically
$M=\Pow(\Sigma^*\times \argument)$, a monad that contains the functor
$\Pow(\Sigma\times \argument)$ via an obvious natural
transformation~$\alpha$. Defining $\gamma^{(n)}$ as the $n$-fold
iteration of the morphism $\alpha\gamma$ in the Kleisli category of
$M$, we have $\gamma^{(0)}(s_0)=\{(\epsilon,s_0)\}$,
$\gamma^{(1)}(s_0)=\{(a,s_{10}),(a,s_{11})\}$,
$\gamma^{(2)}(s_0)=\{(ab,s_{20}),(ac,s_{21})\}$, and
$\gamma^{(3)}(s_0)=\emptyset$. Forgetting the state component of the
pretraces in these sets amounts to postcomposing with $M!$, where $!$
is the unique map into $1=\{*\}$.
These considerations lead to the following definitions.
\begin{defn} \label{defn:trace-sem}
  A \emph{trace semantics} for a functor $G$ is a natural
  transformation $\alpha:G\to M$ into a monad $M$, the \emph{global
    monad}. Given such an $\alpha$ and a $G$-coalgebra $\gamma:X\to
  GX$, we define the \emph{iterations} $\gamma^{(n)}:X\to MX$ of
  $\gamma$, for $n\ge 0$, inductively by
  \begin{equation*}
    \gamma^{(0)}=\eta_X\qquad \gamma^{(n+1)}=(\alpha\gamma)^*\gamma^{(n)}
  \end{equation*}
  where the unit $\eta$ and the Kleisli star $*$ are those of $M$ (in
  particular $\gamma^{(1)}=¸\alpha\gamma$). Then the
  \emph{$\alpha$-trace sequence} of a state $x\in X$ is the sequence
  \begin{equation*}
    T^\alpha_\gamma(x)=(M!\gamma^{(n)}(x))_{n<\omega},
  \end{equation*}
  with $!$ denoting the unique map $X\to 1$ as above. Two states $x$
  and $y$ in $G$-coalgebras $\gamma:X\to GX$ and $\delta:Y\to GY$,
  respectively, are \emph{$\alpha$-trace equivalent} if
  \begin{equation*}
    T^\alpha_\gamma(x)=T^\alpha_\delta(y).
  \end{equation*}
  (Although we use an element-based formulation for readability, this
  definition clearly does make sense over arbitrary complete base
  categories.)
\end{defn}
Of course, one shows by induction over $n$ that
\begin{equation}\label{eq:gamman-alt}
  \gamma^{n+1}=(\gamma^{(n)})^*\alpha\gamma\quad\text{for all $n<\omega$}.
\end{equation}
We first note that the trace sequence factors through the initial
$\omega$-segment of the terminal sequence. Recall from
Section~\ref{sec:prelim} that a $G$-coalgebra $\gamma$ induces a cone
$(\gamma_n)$ into the final sequence. 

\begin{lem}\label{lem:factor}
  Let $\alpha:G\to M$ be a trace semantics for $G$, and define
  natural transformations $\alpha_n:G^n\to M$ for
  $n<\omega$ recursively by
  $
  \alpha_0=\eta$ and $\alpha_{n+1}=\mu\alpha
  G\alpha_n$.
  If $\gamma$
  is a $G$-coalgebra, then
  \begin{equation*}
    M!\gamma^{(n)} = \alpha_n\gamma_n\qquad\text{for all
      $n<\omega$}
  \end{equation*}
  for all $n \in \omega$.
\end{lem}




\begin{proof} Induction on $n$.

  \emph{$n=0$:} We have $M\bang\gamma^{(0)}=M\bang\eta = \eta\bang =
  \alpha_0\gamma_0$.


  \emph{$n\to n+1$:} We have
  \begin{align*}
    & \alpha_{n+1}\gamma_{n+1}  \\
    & = \mu\alpha G(\alpha_n)G\gamma_n\gamma && 
    \text{(Definitions of $\gamma_{n+1}$, $\alpha_{n+1}$)}\\
    & = \mu\alpha G(M\bang \gamma^{(n)})\gamma &&
    \text{(Inductive hypothesis)}\\
    & = \mu M(M\bang\gamma^{(n)})\alpha\gamma &&
    \text{(Naturality of $\alpha$)}\\
    & = M\bang\mu M\gamma^{(n)}\alpha\gamma &&
    \text{(Naturality of $\mu$)}\\
    & = M\bang(\gamma^{(n)})^*\alpha\gamma \\
    & = M!\gamma^{(n+1)}&&
    \text{(\ref{eq:gamman-alt})}.
  \end{align*}\qed
\end{proof}

\begin{cor}
  Finite-depth behaviourally equivalent states are $\alpha$-trace
  equivalent.
\end{cor}

\begin{rem}
  In most items of related work, stronger assumptions than we make
  here allow for identifying an \emph{object} of traces in a suitable
  category, such as the Kleisli category~\cite{HasuoEA07} or the
  Eilenberg-Moore category~\cite{JacobsEA12,bms13} of a monad that
  forms part of the type functor. In our setting, a similar endeavour
  boils down to characterizing, possibly by means of a limit of a
  suitable diagram, those $\alpha$-trace sequences that are
  \emph{$G$-realizable}, i.e.\ induced by a state in some
  $G$-coalgebra. We do not currently have a general answer for this
  but point out that in a variant of the special case treated in the
  beginning of the section where we take $G$ to be
  $\Pow^*(\Sigma\times\argument)$, with $\Pow^*$ denoting nonempty
  powerset, and $M=\Pow(\Sigma^*\times\argument)$), the set of
  $G$-realizable traces is the limit of the infinite diagram
  
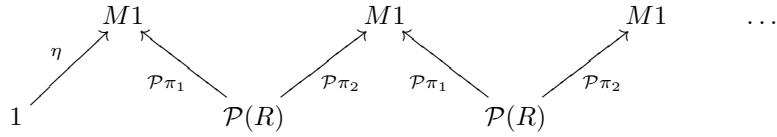
\begin{equation*}
    \xymatrix{
      & M1 & & M1 & & M1 &  \dots\\
      1\ar[ur]^\eta & & \Pow(R)\ar[ul]^{\Pow\pi_1} \ar[ur]_{\Pow\pi_2} & &
    \Pow(R)\ar[ul]^{\Pow\pi_1} \ar[ur]_{\Pow\pi_2}&} 
  \end{equation*}
  where $R$ denotes the immediate prefix relation $R=\{(u,ua)\mid
  u\in\Sigma^*,a\in\Sigma\}$ with projections
  $\pi_1,\pi_2:R\to\Sigma^*$. We expect that this description
  generalizes to cases where $G$ and $M$ have the form $TF$ and
  $TF^*$, respectively, where $T$ is a monad and $F^*$ is the free
  monad over the functor $F$, possibly under additional
  assumptions. In the case at hand, the limit of the diagram is the
  set of all subsets $A$ of $\Sigma^*\times 1\cong\Sigma^*$ that are
  prefix-closed and \emph{extensible} in the sense that for every
  $u\in A$ there exists $a\in\Sigma$ such that $ua\in A$.
\end{rem}

\section{Examples}\label{sec:examples}

We show that various process equivalences are
subsumed under $\alpha$-trace equivalence.

\subsubsection{Finite-depth behavioural equivalence} One pleasant
aspect of $\alpha$-trace equivalence is that it spans, at least for
finitely branching systems, the entire length of the
linear-time-branching-time spectrum, in the sense that even
(finite-depth) behavioural equivalence coincides with $\alpha$-trace
equivalence for a suitable $\alpha$.  This is conveniently
formulated using the following terminology.

\begin{defn}
  We say that an endofunctor $G$ on a category with a terminal object
  $1$ is \emph{non-empty} if $G1$ has a global element.
\end{defn}

\noindent
Non-emptyness of an endofunctor entails that the component of
$\alpha_n$ at $1$ are sections where $\alpha_n$ is as in Lemma
\ref{lem:factor}.

\begin{lem}\label{lem:alphan}
  If $G$ is non-empty and generates an algebraically-free monad $G^*$
  with universal arrow $\alpha$, then $(\alpha_n)_1$ (the component of
  $\alpha_n$ at the terminal object) is a section for every
  $n<\omega$.
\end{lem}
\begin{proof}

  For each set $X$, $G^*X$ is the initial $G+X$-algebra,
  with structure map
  \begin{equation*}
    [\mu\alpha,\eta]:GG^*X+X \to G^*X
  \end{equation*}
  where $\mu$ and $\eta$ are the multiplication and unit of
  $G^*$~\cite{Barr70}.
  By Lambek's lemma, it follows that $[\mu\alpha,\eta]$ is an
  isomorphism. Since both summands of the coproduct $GG^*1+1$ are
  nonempty (for $GG^*1$, this follows from non-emptyness of $G$: we
  obtain a global element of $GG^*1$ by postcomposing a global element
  of $G1$ with $G\eta_1:G1\to GG^*1$), the coproduct injections are
  sections, so we obtain that $\mu\alpha$ and $\eta$ are sections,
  each being the composite of a section with an isomorphism. Using
  (1), it follows by induction that $\alpha_n$ is a section
  for each $n<\omega$. \qed
\end{proof}

\noindent
(Notice that $G$ is non-empty as soon as any $GX$ has a global
element; if the base category is $\Set$, then every functor is
non-empty except the constant functor for $\emptyset$.)
\begin{propn} \label{propn:strong-bisim}
  If $G$ is non-empty and generates an algebraically-free monad via
  $\alpha:G\to G^*$, then $\alpha$-trace equivalence coincides with
  $\omega$-behavioural equivalence.
\end{propn}
\begin{proof}
  Immediate from Lemmas~\ref{lem:factor} and~\ref{lem:alphan}
\qed\end{proof}

\subsubsection{Labelled Transition Systems (LTS)} We provide some additional
details for our initial example: We have $GX=\Pow(\Sigma\times X)$ and
$MX=\Pow(\Sigma^*\times X)$, with $\alpha$ the obvious inclusion. The
monad  $M$ arises from $G$, as we will see later again in \eqref{eq:MTF}, from  a distributive law $\delta_X:\Sigma\times\Pow(X)\to\Pow(\Sigma\times X)$
which maps a pair $(a,S)$ to $\{a\}\times S$. Explicitly, the unit of
$M$ is given by $\eta(x)=\{(\epsilon,x)\}$, and the multiplication by
$\mu(\FA)=\{(uv,x)\mid \exists (u,S)\in\FA.\,(v,x)\in S\}$ for
$\FA\in\Pow(\Sigma^*\times\Pow(\Sigma^*\times X))$. For each $n$ and
each state $x$ in an LTS $\gamma:X\to\Pow(\Sigma\times X)$,
$\gamma^{(n)}(x)$ consists of the pretraces of $x$ of length exactly $n$,
i.e.\
\begin{equation*}
  \gamma^{(n)}(x)=\{(u,y)\mid x\stackrel{u}{\to} y, u\in\Sigma^n\}
\end{equation*}
where $\stackrel{u}{\to}$ denotes the usual extension of the
transition relation to words $u\in\Sigma^*$. Thus, $M!\gamma^{(n)}(x)$
consists of the traces of $x$ of length $n$, i.e.\
$M!\gamma^{(n)}(x)=\{(u,*)\mid x\stackrel{u}{\to},u\in\Sigma^n\}$
(where, as usual, $x\stackrel{u}{\to}$ denotes $\exists
y.x\stackrel{u}{\to}y$). Thus, states $x$ and $y$ are $\alpha$-trace
equivalent iff they are trace equivalent in the usual sense, i.e.\ iff
$\{u\in\Sigma^*\mid x\stackrel{u}{\to}\}=\{u\in\Sigma^*\mid
y\stackrel{u}{\to}\}$. The entire scenario transfers verbatim to the
case of finitely branching LTS, with
$G=\Pow_\omega(\Sigma\times\argument)$ and
$M=\Pow_{<\omega}(\Sigma^*\times\argument)$, where $\Pow_{<\omega}$ denotes
finite powerset.

\subsubsection{LTS with explicit termination} The leading example
treated in related work on coalgebraic trace
semantics~\cite{HasuoEA07,JacobsEA12,bms13} is a variant of LTS with
explicit termination, described as coalgebras for the functor
\begin{equation*}
  \Pow(1+\Sigma\times\argument)\cong 2\times\Pow^\Sigma.
\end{equation*}
A state in an LTS with explicit termination can be seen as a
non-deterministic automaton; this suggests that one might expect the
traces of such a state to be the words accepted by the corresponding
automaton, and this in fact the stance taken in previous
work~\cite{HasuoEA07,JacobsEA12,bms13}; for the sake of distinction, let us
call this form of trace semantics \emph{language semantics}. Staring
at the problem for a moment reveals that language semantics does not
fit directly into our framework: Basically, our definition of trace
sequence assembles the traces via successive iteration of the
coalgebra structure, and remembers the traces reached in each
iteration step. Contrastingly, language semantics will drop a word
from the trace set if it turns out that upon complete execution of the
word, no accepting state is reached -- in $\alpha$-trace semantics, on
the other hand, we will have recorded prefixes of the word on the way,
and our incremental approach does not foresee forgetting these
prefixes. See Section~\ref{sec:related} for a discussion of how
$\alpha$-trace sequences can be further quotiented to obtain language
semantics.

Indeed one might contend that a more natural trace semantics of an LTS
with explicit termination will distinguish two types of traces: those
induced by the plain LTS structure, disregarding acceptance, and those
that additionally end up in accepting states; this is related to the
trace semantics of CSP~\cite{Hoare85}, which distinguishes deadlock
from successful termination $\checkmark$. Such a semantics is
generated by our framework as follows.  As the global monad, we take
$MX=\Pow(\Sigma^*\times(X+1))$ (where we regard $X$ and
$1=\{\checkmark\}$ as subsets of $X+1$), with
$\eta(x)=\{(\epsilon,x)\}$ and
\begin{equation*}
  f^*(S)=\{(uv,b)\mid \exists(u,x)\in
  S\cap(\Sigma^*\times X).\,(v,b)\in f(x)\}\cup (S\cap(\Sigma^*\times 1))
\end{equation*}
for $f:X\to MY$ and $S\in MY$. This is exactly the monad induced by
the distributive law
$\lambda_X:1+\Sigma\times\Pow(X)\to\Pow(1+\Sigma\times X)$ with
$\lambda_X(\checkmark)=\{\checkmark\}$ and $\lambda_X(a,S)=a\times S$
as used by Hasuo et al.~\cite{HasuoEA07}. We embed
$\Pow(1+\Sigma\times\argument)$ into $M$ by the natural
transformation $\alpha$ given by
\begin{equation*}
  \alpha_X(S)=\{(\epsilon,\checkmark)\mid \checkmark\in S\}\cup 
  \{(a,x)\mid (a,x)\in S\}
\end{equation*}
(implicitly converting letters into words in the second part). Then
$M1\cong\Pow(\Sigma^*)^2$ where the first components records accepted
words and the second component non-blocked words; in $\alpha$-trace
sequences, the first component is always contained in the second one,
and increases monotonically over the sequence as the Kleisli star as
defined above always keeps traces that are already accepted. Two
states are $\alpha$-trace equivalent iff they generate the same traces
and the same accepted traces, in the sense discussed above. 

All this is not to say that our framework does not cover the language
semantics of non-deterministic automata. Note that we can impose
w.l.o.g.\ that a non-deterministic automaton never blocks an input
letter -- if a state fails to have an $a$-successor, just add an
$a$-transition into a non-accepting state that loops on all input
letters and has no transitions into other states; this clearly leaves
the language of the automaton unchanged. This restriction amounts to
considering coalgebras for the subfunctor
\begin{equation*}
  G=2\times (\Pow^*)^\Sigma
\end{equation*}
of the functor $\Pow(1+\Sigma\times\argument)$ modelling LTS with
explicit termination, where $\Pow^*$ denotes non-empty powerset. We
embed this functor into the same monad $M$ as above, by restricting
$\alpha:\Pow(1+\Sigma\times\argument)\to M$ to $G$. Calling
$G$-coalgebras \emph{non-blocking non-deterministic automata}, we now
have that \emph{two states in a non-blocking non-deterministic
  automaton are $\alpha$-trace equivalent iff they accept the same
  language}. For a coalgebra $\gamma:X\to GX$, the maps
$\gamma^{(n)}:X\to M1$, of course, still record accepted traces as
well as plain traces, but the plain traces no longer carry any
information: all $\alpha$-trace sequences have the form
$(L_n,\Sigma^n)_{n<\omega}$ (with $L_n\subseteq\Sigma^*$ recording the
accepted words of length at most $n$).

\subsubsection{Probabilistic Transition Systems} Recall that
\emph{generative probabilistic (transition) systems} (for simplicity
without the possibility of deadlock, not to be confused with explicit
termination) are modelled as coalgebras for the functor
$\fD(\Sigma\times\argument)$ where $\fD$ denotes the discrete
distribution functor (i.e.\ $\fD(X)$ is the set of discrete
probability distributions on $X$, and $\fD(f)$ takes image measures
under $f$). That is, each state has a probability distribution over
pairs of actions and successor states. We embed
$\fD(\Sigma\times\argument)$ into the global monad
$MX=\fD(\Sigma^*\times\argument)$ via the natural transformation
$\alpha$ that takes a discrete distribution $\mu$ on $\Sigma\times X$
to the discrete distribution on $\Sigma^*\times X$ that behaves like
$\mu$ on $\Sigma\times X$ (where we see $\Sigma$ as a subset of
$\Sigma^*$) and is $0$ outside $\Sigma\times X$. The unit $\eta$ of
$M$ maps $x\in X$ to the Dirac distribution at $(\epsilon,x)$, and for
$f:X\to MY$, 
\begin{equation*}
  f^*(\mu)(u,y)=\sum_{u=vw,x\in X}\mu(v,x)f(x)(w,y)
\end{equation*}
for all $\mu\in MX$, $(u,y)\in\Sigma^*\times Y$. This is the monad
induced by the canonical distributive law~\cite{HasuoEA07}
$\lambda:\Sigma\times\fD\to \fD(\Sigma\times\argument)$ given by
$\lambda_X(a,\mu)=\delta(a)*\mu$ where $\delta$ forms Dirac measures
and $*$ is product measure. We identify $M1$ with
$\fD(\Sigma^*)$. Given these data, observe that for $\gamma:X\to
\fD(\Sigma\times X)$ and $x\in X$, each distribution
$M!\gamma^{(n)}(x)$ is concentrated at traces of length $n$.

Assume from now on that $\Sigma$ is finite. Recall that the usual
$\sigma$-algebra on the set $\Sigma^\omega$ of infinite words over
$\Sigma$ is generated by the \emph{cones}, i.e.\ the sets
$v{\uparrow}=\{vw\mid w\in\Sigma^\omega\}$, $v\in\Sigma^*$, which (by
finiteness of $\Sigma$) form a semiring of sets. We let states $x$ in
a coalgebra $\gamma:X\to\fD(\Sigma\times X)$ \emph{induce}
distributions $\mu_x$ on $\Sigma^\omega$ via the Hahn-Kolmogorov
theorem, defining a content $\mu(v{\uparrow})$ inductively by
\begin{align*}
  \mu_x(\epsilon{\uparrow})& =1\\
  \mu_x(av{\uparrow}) & = \sum_{x'\in X}\gamma(a,x')\mu_{x'}(v{\uparrow})
\end{align*}
-- a compactness argument, again hinging on finiteness of $\Sigma$,
shows that no cone be written as a countably infinite disjoint union
of cones, so $\mu$ is in fact a pre-measure, i.e.\ $\sigma$-additive.

We note explicitly
\begin{propn}
  States in generative probabilistic systems over a finite alphabet
  $\Sigma$ are $\alpha$-trace equivalent iff they induce the same
  distribution on $\Sigma^\omega$.
\end{propn}
\begin{proof}
  For $v$ a word of length $n$ and $x$ a state in a generative
  probabilistic system, we have
  \begin{equation*}
    \mu_x(v{\uparrow})=(M!\gamma^{(n)}(x))(v). 
  \end{equation*}
\qed\end{proof}
\section{Relation to Other Frameworks}\label{sec:related}

\subsubsection*{Kleisli Liftings} Hasuo et al.~\cite{HasuoEA07} treat
the case where the type functor $G$ has the form $TF$ for a monad $T$
and a finitary endofunctor $F$ on sets. They require that $F$ lifts to
a functor $\bar F$ on the Kleisli category of $T$, which is equivalent
to having a (functor-over-monad) distributive law
\begin{equation*}
  \lambda:FT\to TF.
\end{equation*}
They impose further conditions that include a cppo structure on the
hom-sets of the Kleisli category $\Kl(T)$ of $T$ and ensure that
\begin{itemize}
\item $T\emptyset$ is a singleton, so that $\emptyset$ is a terminal
  object in $\Kl(T)$ (unique Kleisli morphisms into $\emptyset$ of
  course being $\bot$); and
\item the final sequence of $\bar F$ coincides on objects with the
  initial sequence of $F$, and converges to the final $\bar
  F$-coalgebra in $\omega$ steps.
\end{itemize}
The trace semantics of a $TF$-coalgebra is then defined as the unique
Kleisli morphism into the final $\bar F$-coalgebra; in keeping with
distinguishing terminology used in Section~\ref{sec:examples}, we
refer to this as language semantics. Thus, two states in a
$TF$-coalgebra are \emph{language equivalent}, i.e.\ trace equivalent
in the sense of Hasuo et al., iff they map to the same values in the
final sequence of $\bar F$ under the cones induced by the respective
coalgebras. Explicitly: the underlying sets of the final sequence of
$\bar F$ have the form $TF^n\emptyset$, $n<\omega$, and given a
coalgebra $\gamma:X\to TFX$, the canonical cone $(\bar\gamma_n:X\to
TF^n\emptyset)_{n<\omega}$ is defined recursively by $\gamma_0=\bot$
and
\begin{equation*}
\bar\gamma_{n+1}=\xymatrix{X\ar[r]^-\gamma & TFX\ar[r]^-{TF\bar\gamma_n}
  & TFTF^n\emptyset \ar[r]^{T\lambda} & TTF^{n+1}\emptyset
  \ar[r]^{\mu} & TF^{n+1}\emptyset.}
\end{equation*}
Now the distributive law $\lambda$ induces a monad structure on the
functor 
\begin{equation}\label{eq:MTF}
M=TF^*,
\end{equation} where $F^*$ denotes the (algebraically-)free monad
on $F$ (cf.\ Section~\ref{sec:examples}), and we have a natural
transformation $\alpha:TF\to M$, so that the situation fits our
current framework. The sets $TF^nX$ embed into $MX$, so that the
objects in the final sequence of $\bar F$ can be seen as living in
$M0$. The definition of $\bar\gamma_{n+1}$ is then seen to be just an
explicit form of Kleisli composition in $M$; that is, we can, for
purposes of language equivalence, replace the $\bar\gamma_n$ with maps
$\tilde\gamma_n:X\to M0$ defined recursively by
\begin{equation*}
  \tilde\gamma_0=\bot\qquad\tilde\gamma_{n+1}=\tilde\gamma_n^*\alpha\gamma
\end{equation*}
where the Kleisli star is that of $M$. Comparing with
(\ref{eq:gamman-alt}), we see that the only difference with the
definition of $\gamma^{(n)}$ is in the base of the recursion:
$\gamma^{(0)}=\eta_X$ . Noting moreover that 
\begin{equation*}
  \bot^*M!\eta_X=\bot^*\eta !=\bot!=\bot,
\end{equation*}
we obtain
\begin{equation*}
  \tilde\gamma_n=\bot^*M!\gamma^{(n)}.
\end{equation*}
(Kissig and Kurz~\cite{KissigKurz10} use a very similar definition in
a more general setting that in particular, for non-commutative $T$,
does not restrict $T\emptyset$ to be a singleton, and instead assume
some distinguished element $e\in T\emptyset$. They then put
$\tilde\gamma_0=\lambda x.\,e$; the comparison with our framework is
then entirely analogous.)

 Summing up, \emph{language equivalence is induced from
  $\alpha$-trace equivalence by postcomposing $\alpha$-trace sequences
  with $\bot^*:M1\to M0$}. Intuitively, this means that any
information tied to poststates in a pretrace is erased in language
equivalence, as opposed to just forgetting the poststate itself in
$\alpha$-trace equivalence. An example of this phenomenon are LTS with
explicit termination as discussed in
Section~\ref{sec:examples}. Moreover, this observation elucidates why
language equivalence becomes trivial in cases without explicit
termination, such as standard LTS: here, all traces are tied to
poststates and hence are erased when postcomposing with
$\bot^*$. (This is also easily seen directly~\cite{HasuoEA07}: without
explicit termination, e.g.\ $F=\Sigma\times\argument$, one typically
has $F\emptyset=\emptyset$ so that the final $\bar F$-coalgebra is
trivial in the Kleisli category of $M$.)

\subsubsection*{Eilenberg-Moore Liftings} An alternative route to
final objects for trace semantics was first suggested by the
generalized powerset construction of Silva et
al.~\cite{bbrs_fsttcs} and explicitly formulated in~\cite{bms13} (see
also Jacobs et al.~\cite{JacobsEA12} where this is compared to the
semantics given by Kleisli liftings). In this approach one considers liftings of functors to
Eilenberg-Moore categories in lieu of Kleisli categories. The setup
applies to functors of the form $G=FT$ where $F$ is an endofunctor and
$T$ is a monad on a base category $\BC$. It is based on assuming a
final $F$-coalgebra $Z$ and a (functor-over-monad) distributive law
\begin{equation*}
  \rho: TF\to FT.
\end{equation*}
Under these assumptions, $F$ lifts to an endofunctor $\hat F$ on the
Eilenberg-Moore category $\BC^T$ of $T$, and the free-algebra functor
$\BC\to\BC^T$ lifts to a functor $D$ from $FT$-coalgebras to $\hat
F$-coalgebras, which can be seen as a generalized powerset
construction. Explicitly, $D(\gamma)=F\mu^T_X\rho_{TX}T\gamma$ for
$\gamma:X\to FTX$, where $\mu^T$ denotes the multiplication of $T$. In
other words, $D(\gamma): TX \to FTX$ is the unique $T$-algebra
morphism with $D(\gamma) \cdot \eta^T_X = \gamma$. Moreover, $\hat F$ has
a final coalgebra with carrier $Z$. The \emph{extension semantics}
(i.e.\ trace semantics obtained via the powerset extension) of an
$FT$-coalgebra $\gamma:X\to FTX$ is then obtained by first applying
$D$ to $\gamma$, obtaining a $\hat F$-coalgebra with carrier $TX$ and
hence a $\hat F$-coalgebra map $TX\to Z$, and finally precomposing
with $\eta^T_X:X\to TX$ where $\eta^T$ denotes the unit of $T$.

In order to compare this with our framework, in which we currently
consider only finite iterates of the given coalgebra, we need to
assume that $F$-behavioural equivalence coincides with finite-depth
behavioural equivalence; this is ensured e.g.~by assuming that $F$ is a
finitary endofunctor on $\Set$. In this case,
two states have the same extension semantics iff they induce the same
values in the first $\omega$ steps of the final sequence of $\hat F$,
whose carriers coincide with the final sequence of $F$. Combining the
definition of $D\gamma$ for a coalgebra $\gamma:X\to FTX$ with the
usual construction of the canonical cone for $D\gamma$, which we
denote by $\bar\gamma_n:TX\to F^n1$ for distinction from the canonical
cone of $\gamma$ in the final sequence of $FT$, we obtain that
$\bar\gamma_n$ is recursively defined by
\begin{align*}
  \bar\gamma_0 & =  \bang_{TX}:TX\to 1\\
  \bar\gamma_{n+1} & = F\bar\gamma_nT\gamma\rho F\mu^T.
\end{align*}
Now let us also assume that $T$ is a finitary monad on $\Set$. Then
$\Set^T$ is a locally finitely presentable category, and since the
forgetful functor to $\Set$ creates filtered colimits, we see that the
lifting $\hat F$ is finitary on $\Set^T$. Hence free $\hat F$-algebras
exists, which implies that we have the adjunction on the right below
\[
\xymatrix@1{
  \Set 
  \ar@<4pt>[r] 
  \ar@{}[r]|-\perp
  & 
  \Set^T \ar@<4pt>[l] 
  \ar@<4pt>[r] 
  \ar@{}[r]|-\perp
  &
  \mathsf{Alg}\,\hat F
  \ar@<4pt>[l] 
},
\]
and the adjunction on the left is the canonical one. We define $M$ to
be the monad of the composed adjunction; it assigns to a set $X$ the
underlying set $\hat F^* TX$ of a free $\hat F$-algebra on the free
$T$-algebra $TX$; here $\hat F^*$ denotes the free monad on $\hat F$
(notice that this is not in general a lifing of the free monad on $F$
to $\Set^T$). Intuitively, $M$ is defined by forming the disjoint
union of the algebraic theories associated to $T$ and $F$,
respectively, and then imposing the distributive law between the
operations of $T$ and $F$ embodied by $\rho$. In the following we
shall denote the unit and multiplication of $\hat F^*$ by $\hat\eta$
and $\hat\mu$, respectively. We also write $\hat\varphi_X: \hat F \hat
F^*X \to \hat F^* X$ for the structures of the free $\hat F$-algebras
and note that these yield a natural transformation $\hat\varphi$.

Now denote by $\hat\kappa: \hat F \to \hat F^*$ the universal natural
transformation into the free monad; it is easy to see that $\hat\kappa
= \hat \varphi \cdot \hat F\hat \eta$. Then it follows that $\alpha =
\hat\kappa T$ yields a natural transformation from $FT$ to $M$ (on
$\Set$). Let us further recall that there exist canonical natural
transformations $\hat\beta^n: \hat F^n \to \hat F^*$ defined
inductively by
\[
\hat\beta^0 = (\xymatrix@1{\Id \ar[r]^-{\hat\eta} & \hat F^*})
\qquad\text{and}\qquad
\hat\beta^{n+1}= (\xymatrix@1{
\hat F^{n+1} = \hat F\hat F^n \ar[r]^-{\hat F \hat\beta^n} & \hat
F\hat F^* \ar[r]^-{\hat\varphi} & \hat F^*
}).
\]
We can assume w.l.o.g.\ that $F$ preserves monos (hence, so does $\hat
F$ since monos in $\Set^T$ are precisely injective $T$-algebra
homomorphisms) and that coproduct injections are monic in
$\Set^T$. Then an easy induction shows that the $\beta^n$ are monic,
too. (One uses that $[\hat\eta,\hat\phi]:\Id + \hat F \hat
F^*\cong\hat F^*$.) This implies that for testing 
equivalence in the extension semantics we can replace $\bar\gamma_n$ with
\begin{equation*}
  \hat\gamma_n=\beta^n_1\cdot \bar\gamma_n:TX\to \hat F^*1.
\end{equation*}
We are now ready to state the semantic comparison result:
\begin{thm}\label{thm:jss}
  Let $F$ be a finitary endofunctor, and let $T$ be a finitary monad,
  both on $\Set$. Further let $\rho:TF\to FT$ be a functor-over-monad
  distributive law. Then two states in $FT$-coalgebras are equivalent
  under the extension semantics iff for $\alpha:FT\to M$ as
  given above, their $\alpha$-trace sequences are identified under
  componentwise postcomposition with $\hat F^*\bang_{T1}$. That is, in
  the above notation,
  \begin{equation}
    \label{eq:s}
    \hat\gamma_n\cdot \eta^T_X=\hat F^*\bang_{T1}\cdot M\bang_X\cdot\gamma^{(n)}.
  \end{equation}
\end{thm}
\begin{proof}
  We first recall how the Kleisli extension $f \mapsto f^*$ for the
  monad $M$ is obtained. Given $f: X \to MY$ one first extends this to
  the unique $T$-algebra morphism $f^\sharp: TX \to MY$ with $f^\sharp
  \cdot \eta^T_X = f$ (i.\,e.~one applies the Kleisli extension of
  $T$). Then one obtains $f^*: MX = \hat F^*TX \to \hat F^* TY = MY$ as the unique $\hat
  F$-algebra morphism with $f^* \cdot \hat \eta_{TX} =
  f^\sharp$. Notice that in this notation we have $D(\gamma) =
  \gamma^\sharp$ and that the inductive step of the definition on
  $\bar\gamma_n$ can be written as $\bar\gamma_{n+1} = \hat F
  \bar\gamma_n \cdot \gamma^{\sharp}: TX \to \hat F^n 1$. Observe
  further that, since $\hat\gamma_n$, $\hat F^*\bang_{T1}$ and $M\bang$ are $T$-algebra homomorphisms, \refeq{eq:s} is equivalent to 
  \begin{equation}
    \label{eq:ind}
    \hat\gamma_n = \hat F^*\bang_{T1}\cdot M\bang_X\cdot (\gamma^{(n)})^\sharp.
  \end{equation}
  We now prove~\refeq{eq:s} by induction on $n$. For the base case $n
  = 0$ we have:
  \[
  \begin{array}{rcl@{\quad}p{5cm}}
    \hat F^*\bang_{T1}\cdot M\bang_X\cdot \gamma^{(0)}
    & = & \hat F^* \bang_{T1}\cdot \hat F^* T\bang_X \cdot \eta^M_X 
    & $M = \hat F^* T$ and def.~of $\gamma^{(0)}$
    \\
    & = & \hat F^*\bang_{T1}\cdot \hat F^* T\bang_X \cdot \hat\eta_{TX} \cdot \eta^T_X & since $\eta^M = \hat \eta T \cdot \eta^T$ 
    \\
    & = & \hat \eta_1 \cdot \bang_{T1}\cdot T\bang_X \cdot \eta^T_X & naturality of $\hat \eta$ \\
    & = & \hat \eta_1 \cdot \bang_{TX} \cdot \eta^T_X & uniqueness of $\bang_{TX}$ \\
    & = & \hat\beta^0_1 \cdot \bar\gamma_0 \cdot \eta^T_X & def.~of $\hat\beta^0$ and $\bar\gamma_0$ \\
    & = & \hat \gamma_0 \cdot \eta^T_X & def.~of $\hat \gamma_0$.
  \end{array}
  \]
  For the induction step we compute:
  \[
  \begin{array}{rcl@{\quad}p{5cm}}
    \multicolumn{4}{l}{\hat F^*\bang_{T1}\cdot M\bang_X\cdot \gamma^{(n+1)}} \\
    \quad
    & = & \hat F^*\bang_{T1}\cdot \hat F^* T\bang_X\cdot (\gamma^{(n)})^* \cdot \alpha_X \cdot \gamma
    & $M = \hat F^* T$ and def.~of $\gamma^{(n+1)}$ \\
    & = & \hat F^*\bang_{T1}\cdot \hat F^* T\bang_X\cdot (\gamma^{(n)})^* \cdot \hat\varphi_X \cdot \hat F\hat\eta_{TX} \cdot \gamma & def.~of $\alpha$ \\
    & = & \hat\varphi_1 \cdot \hat F\hat F^*\bang_{T1}\cdot \hat F \hat F^* T\bang_X\cdot \hat F (\gamma^{(n)})^*\cdot \hat F\hat\eta_{TX} \cdot \gamma 
    & $\hat F$-algebra morphisms \\
    & = & \hat\varphi_1 \cdot \hat F\hat F^*\bang_{T1}\cdot \hat F \hat F^* T\bang_X\cdot \hat F (\gamma^{(n)})^\sharp \cdot \gamma & def.~of $(-)^*$ \\
    & = & \hat\varphi_1 \cdot \hat F \hat\gamma_n \cdot \gamma & induction hypothesis~\refeq{eq:ind} \\
    & = & \hat\varphi_1 \cdot \hat F \hat\beta^n_1 \cdot \hat F \bar \gamma_n \cdot \gamma & def.~of $\hat\gamma^n$ \\
    & = & \hat\beta^{n+1}_1 \cdot F \bar \gamma_n \cdot \gamma & def.~of $\beta^{n+1}$ \\
    & = & \hat\beta^{n+1}_1 \cdot F \bar \gamma_n \cdot \gamma^\sharp \cdot \eta^T_X & $(-)^\sharp$ Kleisli extension \\
    & = & \hat\beta^{n+1}_1 \cdot \bar\gamma_{n+1}\cdot \eta^T_X & def.~of $\bar\gamma_{n+1}$ \\
    & = & \hat\gamma_{n+1} \cdot \eta^T_X & def.~of $\hat\gamma_{n+1}$.
  \end{array}
  \vspace*{-17pt}
  \]
  \qed
\end{proof}
In the base example in work on extension semantics
\cite{JacobsEA12,bms13}, the case of non-deterministic automata
understood as coalgebras of the form $\gamma: X \to 2 \times
\Pow(X)^\Sigma$, the situation is as follows. The extension semantics
of $\gamma$~\cite[Section~5.1]{JacobsEA12}
yields a map $\mathit{tr}: X \to \Pow(\Sigma^*)$ that maps each state
$x \in X$ to the language accepted by the automaton with starting
state $x$.

To understand the above theorem in terms of this concrete example, we
fix $FX = 2 \times X^\Sigma$ and $TX = \Powf(X)$ (to ensure
finitarity). Understood as an algebraic signature, $F$ can be
represented by two $\Sigma$-ary function symbols $\yes$ and $\no$. The
monad $M=\hat F^*T$ has these operations and those of $\Powf$, i.e.\
the join semilattice operations, which we write using set notation;
the distributive law $\rho$ allows us to distribute joins over $\yes$
and $\no$, favouring $\yes$ over $\no$ to reflect the acceptance
condition of (existential) non-deterministic automata. The trace
semantics $\alpha_X: FTX \to MX$ embeds flat terms, i.e.\ terms of the
form $\yes( (U_a)_{a \in \Sigma})$ or $\no( (U_a)_{a \in \Sigma}) \in
FTX$ (with $U_a \in \Pow(X)$), into general (non-flat) terms. 
Every step in the construction of $\gamma^{n}(c)$ puts a flat term on
top of terms constructed in the previous step, and then distributes
$T$-operations (joins) over their arguments as indicated. Therefore,
the terms $\gamma^{(n)}(c)$ are terms of uniform depth in the
$F$-operations over sets of variables, i.e.\ they are elements of
$F^nTC$. For the alphabet $\Sigma = \lbrace 0, 1 \rbrace$, a typical
component of the trace sequence $T^\alpha_\gamma(c)$, i.e. $M!_X 
\gamma^{(n)}(c)$ for some $n$ can be visualised as a tree
like the one on the left:
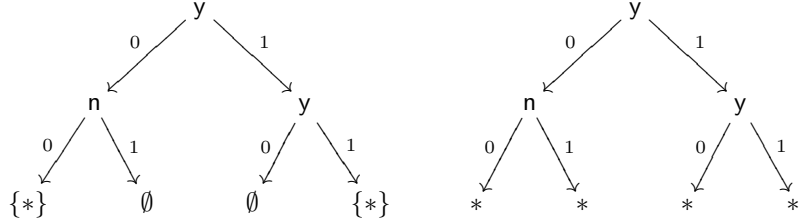
\[
  \xymatrix@C=2ex{ 
  & & & & \yes \ar[dll]_0 \ar[drr]^1 & & & &   & & & & \yes \ar[dll]_0 \ar[drr]^1 & & & &\\
  & &\no \ar[dl]_0 \ar[dr]^1  & & & & \yes \ar[dl]_0 \ar[dr]^1 & &    & &\no \ar[dl]_0 \ar[dr]^1  & & & & \yes \ar[dl]_0 \ar[dr]^1 & &    \\
  & \lbrace \ast \rbrace  & & \emptyset & & \emptyset & &  \lbrace
  \ast \rbrace   &  & \ast  & &  \ast   & & \ast & &  \ast. 
  }
\]
This tree conveys the information that the empty word $\epsilon$ and
the word $1$ lead to final states (i.e.\ are accepted in the sense of
language semantics), and additionally that $00$ and $11$ are not
blocked; generally, the $\alpha$-trace sequence records at each stage
which words are accepted and additionally which words can be executed
without deadlock. The tree on the right is then obtained by applying
$\hat F^*!_{T1}$. This erases the information on non-blocked words, so
that only the information that $\epsilon$ and $1$ are accepted
remains; this yields the extension semantics~\cite{JacobsEA12,bms13},
i.e.\ language semantics of the automaton, as formally stated in
Theorem~\ref{thm:jss}. As noted already in Section~\ref{sec:examples},
if we move to non-blocking non-deterministic automata, then
$\alpha$-trace equivalence coincides directly with language
equivalence -- note that in this case, $T$ is non-empty powerset, so
that $!_{T1}$ is a bijection, i.e.\ postcomposing the $\alpha$-trace
sequence with $\hat F^*!_{T1}$ does not lose information. Informally,
this is clear as non-acceptance of words due to deadlock never happens
in a non-blocking nondeterministic automaton.

\subsubsection*{Fixpoint Definitions} Trace semantics, and
associated linear-time logics, are also considered in
\cite{Cirstea:2014:CAL}. The framework considered in \emph{op.cit.}
is similar to that of \cite{HasuoEA07} in that it applies to systems
of type $X \to TFX$ where $T$ is a monad (that describes the
branching) and $F$ a polynomial
endofunctor (modelling the traces). The monad $T$ is required to be
commutative and partially additive, thus inducing a partial additive
semiring structure on $T1$. In the examples of interest, one
recovers the monad $T$ as induced  by this semiring structure. 

Given a system $(X, f: X \to TFX)$, trace semantics then arises as a
$T1$-valued relation $R: X \times Z \to T1$ where $Z = \nu F$ is the
final coalgebra of the functor $F$ defining traces. For this to be
well-defined, one additionally requires that the semiring $T1$ has
suprema of chains, with order defined in the standard way. 

The crucial difference to our approach is that trace semantics is
defined \emph{coinductively} on the \emph{infinite unfolding} of the
functor $F$ defining the shape of traces, whereas our definition is
\emph{inductive} and based on \emph{finite unfoldings}. 

The difference becomes apparent when looking at examples. For
labelled transition systems $X \to \Pow(A \times X)$, the trace
semantics of \emph{op.cit.} is a function $X \to \Pow(A^w)$ that
maps $x$ to the set of maximal traces, and two states are trace
equivalent if they have the same set of \emph{infinite} traces. This
contrasts with our treatment where equivalent states have the same
\emph{finite} traces. Similarly, for generative probabilistic
systems, i.e. systems of shape $X \to \fD(A \times X)$ where $\fD$ is
the discrete distributions functor, \emph{op.cit.} the trace
semantics obtained in \emph{op.cit.} associates probabilities to
maximal (infinite) traces whereas our treatment is centered around
probabilities of finite prefixes.  In summary, the main conceptual
difference 
between \cite{Cirstea:2014:CAL} and our approach is that between
infinite and finite traces. Technically, this difference is manifest
in the coinductive definition of \emph{op.cit.} whereas our approach
defines traces inductively.

\section{Conclusions}
One of the main important aspects of the general theory of coalgebra
is a uniform  theory of strong bisimulation. In coalgebraic terms,
strong bisimulation is a simple concept, readily defined, supports a
rich theory and instantiates to the natural and known notions for
concretely given transition types. Instead of re-establishing facts
about strong bisimulation on a case-by-case basis, separately for
each type of transition system, the coalgebraic approach
provides a general theory of which specific results for concretely
given systems are mere instances: a coalgebraic success story.

The question about whether a similar success story for trace
equivalence can also be told in a coalgebraic setting has been the
subject of numerous papers (discussed in the previous section in
detail) but has so far not received a satisfactory answer. 

One of the reasons why trace semantics has so far been a more
elusive concept is the fact that -- even for concretely given
systems such as labelled transition systems with explicit
termination -- there are many, equally natural, formulations of
trace equivalence. This suggests that trace equivalence, by its very
nature, cannot be captured by one general definition, but needs an
additional parameter that defines the precise nature of traces one
wants to capture.

In contrast to other approaches in the literature, we account for
this fact by parametrising trace semantics by an embedding of a
functor (that defines the coalgebraic type of system under
consideration) into a monad (that allows us to sequence
transitions). As a consequence, our definition is more flexible, and
subsumes existing notions. Conceptually speaking, this manifests
itself in the fact that other approaches impose various technical
conditions like order enrichment or partial additivity of a monad
that are geared towards capturing a \emph{specific} notion of trace
equivalence, whereas our definition is parametrised to capture the
entire range of the linear-time branching-time spectrum. This is
evidenced by Proposition \ref{propn:strong-bisim} that shows that
(even) strong bisimulation is a specific instance of our
parameterised definition. 

Technically, we have presented a simplified notion of a semantics of
finite traces for coalgebras. This novel account allows us to deal
with new examples and subsumes previous proposals of a semantics of
finite traces. Important points for future work include a
generalisation to behavioural preorders, as well as appropriate logics
that characterise these preorders and ensuing equivalences.

\bibliographystyle{myabbrv}
\bibliography{coalgml}

\end{document}